\documentclass[12pt,onecolumn, draftcls]{IEEEtran}
\usepackage[latin1]{inputenc}
\usepackage{graphicx}
\usepackage{amssymb,amsmath,amsfonts,amsthm}
\usepackage{algorithm,algorithmic}
\usepackage{cite}
\usepackage{float}
\usepackage{graphics}
\usepackage{subfigure}
\usepackage{xspace}
\usepackage[usenames,dvipsnames]{color}
\usepackage{epsfig, hyperref}

\newtheorem{theorem}{Theorem}
\newtheorem{remark}{Remark}
\newtheorem{proposition}{Proposition}

\begin{document}

\title{Opportunistic Beamforming with Wireless Powered $1$-bit Feedback Through Rectenna Array}

\author{Ioannis Krikidis,~\IEEEmembership{Senior Member,~IEEE}
\thanks{Copyright (c) 2015 IEEE. Personal use of this material is permitted. However, permission to use this material for any other purposes must be obtained from the IEEE by sending a request to pubs-permissions@ieee.org.}
\thanks{I. Krikidis is with the Department of Electrical and Computer Engineering, Faculty of Engineering, University of Cyprus, Nicosia 1678 (E-mail: {\sf krikidis@ucy.ac.cy}).}
\thanks{This work was funded by the Research Promotion Foundation, Cyprus, under the Project FUPLEX with Pr. No. CY-IL/0114/02.}

}

\maketitle

\begin{abstract}
This letter deals with the opportunistic beamforming (OBF) scheme for multi-antenna downlink with spatial randomness. In contrast to conventional OBF, the terminals return only $1$-bit feedback, which is powered by wireless power transfer through a rectenna array. We study two fundamental topologies for the combination of the rectenna elements; the direct-current combiner and the radio-frequency combiner. The beam outage probability is derived in closed form for both combination schemes, by using high order statistics and stochastic geometry.               
\end{abstract}


\section{Introduction}

\IEEEPARstart{O}{pportunistic} beamforming (OBF) is a robust communication tool, which exploits multi-user diversity  to achieve full multiplexing gain \cite{SHA, PUN}. OBF requires a continuous feedback mechanism of the achieved beam signal-to-interference-plus-noise ratio (SINR) in order to to assign the orthonormal beams to the corresponding terminals.  Although conventional studies assume homogeneous networks without path-loss effects,  the authors in \cite{SAM} study OBF for networks with location randomness and perfect SINR feedback.          

On the other hand, wireless powered communication (WPC) is a new energy solution for the future autonomous and self-sustainable networks \cite{ZHA,JU,KRI}. It refers to terminals that have the capability to power their operations by the received electromagnetic radiation.  The fundamental block for the implementation of this technology is the rectifying-antenna (rectenna), which is a diode-based circuit that converts the radio-frequency (RF) signals to direct-current (DC) voltage \cite{XIA}. Typically, a single rectenna is not sufficient to support reliable terminal operation. Alternatively, properly interconnecting several rectennas could ensure sufficient rectification \cite{VOL}. This interconnection can be done in the DC or the RF domain; the DC-combiner requires a rectification circuit at each antenna element, while the RF-combiner corresponds to a single rectification circuit.    

In this letter, we study these two fundamental rectenna array configurations for the multi-antenna downlink channel with OBF and spatial randomness. The terminals power their feedback channel  from a power beacon (PB) that broadcasts energy using a separate frequency band. In order to respect the doubly near-far problem associated with WPC \cite{JU}, terminals can only return $1$-bit feedback, which shows the outage status for a preassigned beam \cite{SIM}. The beam outage probability is derived for both rectenna array configuration, by using high order statistics and stochastic geometry tools. To  the  best  of  authors'  knowledge,  the  current  letter  is  the  first  to analytically investigate  OBF with $1$-bit feedback and spatial randomness as well as the impact  of  rectenna array configurations on WPC. 

\section{System model}

We consider a single-cell WPC scenario with multiple randomly deployed terminals, where the coverage area is modeled as a disc of radius $\rho$. An access point (AP) and a PB, operating in different frequencies, are co-located at the origin of the disc with an exclusion zone of radius $\xi$ around them \cite{XIA}; $\mathcal{B}$ denotes the coverage area of interest. The location of the terminals is modeled as a homogeneous Poisson point process (PPP) $\Phi$ with intensity $\lambda$. The AP and the terminals are equipped with $M$ and $N$ antennas, respectively, while the PB has a single transmit antenna; all antennas are omnidirectional. The terminals have wireless power transfer (WPT) capabilities and can harvest energy from the PB's transmitted signals through a rectenna array configuration with $L$ elements. Time is slotted and the RF energy harvested cannot be stored for future use (batteryless architecture i.e., the energy harvested during the $t$-th slot is immediately used e.g., \cite{JU,NAS1,DING}).  $1$-bit feedback is available from the terminals to AP; a random single antenna is used for the transmission and the reception  of the feedback channel (more sophisticated antenna selection schemes cannot be applied due to the considered unfaded uplink channel model).

\subsection{Channel model}
We assume that wireless links suffer from both  small-scale block fading and large-scale path-loss effects. The fading is Rayleigh distributed so the power of the channel fading is an exponential random variable with unit variance; $h_{k,i,j}$ is the channel coefficient for the link between the $k$-th AP's transmit antenna and the  $j$-th receive antenna for the $i$-th terminal;  $g_{i,j}\equiv |g_{i,j}|e^{\jmath \theta_{i,j}}$ is the channel coefficient for the link between the PB and the $j$-th rectenna of the $i$-th terminal. The received power is proportional to $d_i^{-\alpha}$ where $d_i$ is the Euclidean distance between the AP/PB and the $i$-th terminal, $\alpha>2$ denotes the path-loss exponent. In addition, all wireless links exhibit additive white Gaussian noise (AWGN) with variance $\sigma^2$; $n_{i,j}$ denotes the AWGN at the $j$-th receive antenna for the $i$-terminal.  We assume an unfaded AWGN channel for the feedback link which only suffers from path-loss attenuation \cite{KOB,LIU}; this model highlights the impact of the doubly near-far problem on the achieved performance and simplifies our analysis. 

\subsection{Data communication}
The AP applies an OBF strategy and in each time slot it serves
selected users; to do that, it generates $M$ isotropic distributed random orthonormal vectors $\{\pmb{u}_1,\ldots,\pmb{u}_M \}$ with $\pmb{u}_m\in \mathbb{C}^{M\times 1}$, which  represent the beams that are used in order to transmit the $M$ information streams. By omitting time index and carriers, the baseband-equivalent transmitted signal is given by
\begin{align}
\pmb{v}=\sum_{m=1}^{M} \pmb{u}_m s_m,
\end{align}
where $\mathbb{E}[||\pmb{v}||^2]=M$, and $s_m$ is the $m$-th transmitted symbol. The signal received at the $j$-th antenna of the $i$-th terminal is given by
\begin{align}
r_{i,j}=\sqrt{P_{t}d_i^{-\alpha}}\pmb{h}_{i,j}^{T}\pmb{v}+n_{i,j},  
\end{align}
where $\pmb{h}_{i,j}^T=[h_{1,i,j},\ldots, h_{M,i,j}]$, $P_t$ denotes the AP's transmitted power. The SINR for the $l$-th beam at the $j$-th receive antenna of the $i$-th terminal is equal to
\begin{align}
\Gamma_{i,j,l}=\frac{|\pmb{h}_{i,j}^T\pmb{u}_l|^2}{\frac{d_i^\alpha}{P_t}+\sum_{m\neq l}^M |\pmb{h}_{i,j}^T \pmb{u}_m|^2}.
\end{align}
Each terminal employs a selection combiner (SC) \cite{PUN} in order to keep the complexity low (i.e., $1$ RF chain) and therefore the achieved SINR for the $l$-th beam at the $i$-th terminal is given by 
\begin{align}
\Gamma_{i,l}=\max_{1\leq j \leq N}\Gamma_{i,j,l}.
\end{align}

\subsection{WPT operation}
The PB operates in a separate frequency band in order to avoid interference with the communication links \cite{XIA}. The transmitted RF signal at the $t$-th time slot is given by 
\begin{align}
s(t)\!=\!\sqrt{2P_h}\Re\!\left\{\!x(t)e^{\jmath 2\pi ft} \!\right\}\!\!=\!\!\sqrt{2P_h}\Re\!\left\{\!e^{\jmath [2\pi ft+\arg{x(t)}]}\!\right\},
\end{align}
where $\Re\{z\}$ denotes the real part of $z$, $P_h=\mathbb{E}[s^2(t)]$ is the PB's transmit power, $f$ denotes the carrier frequency, and $x(t)$ is a modulated energy signal with $|x(t)|^2=1$. The received signal at the $j$-th antenna of the $i$-th terminal is given by 
\begin{align}
y_{i,j}(t)&\!=\!\sqrt{2P_hd_i^{-\alpha}}|g_{i,j}(t)| \Re \left\{e^{\jmath [2\pi ft+\arg{x(t)}+\theta_{i,j}(t)]} \right\} \nonumber \\
\;\;\;&=\sqrt{2P_h d_i^{-\alpha}}|g_{i,j}(t)|\cos\left(2\pi ft+\arg{x(t)}+\theta_{i,j}(t) \right),
\end{align}
where WPT from AWGN is considered negligible and thus can be ignored \cite{ZHA,JU}. 

\begin{figure}[t]
\centering
\includegraphics[width=\linewidth]{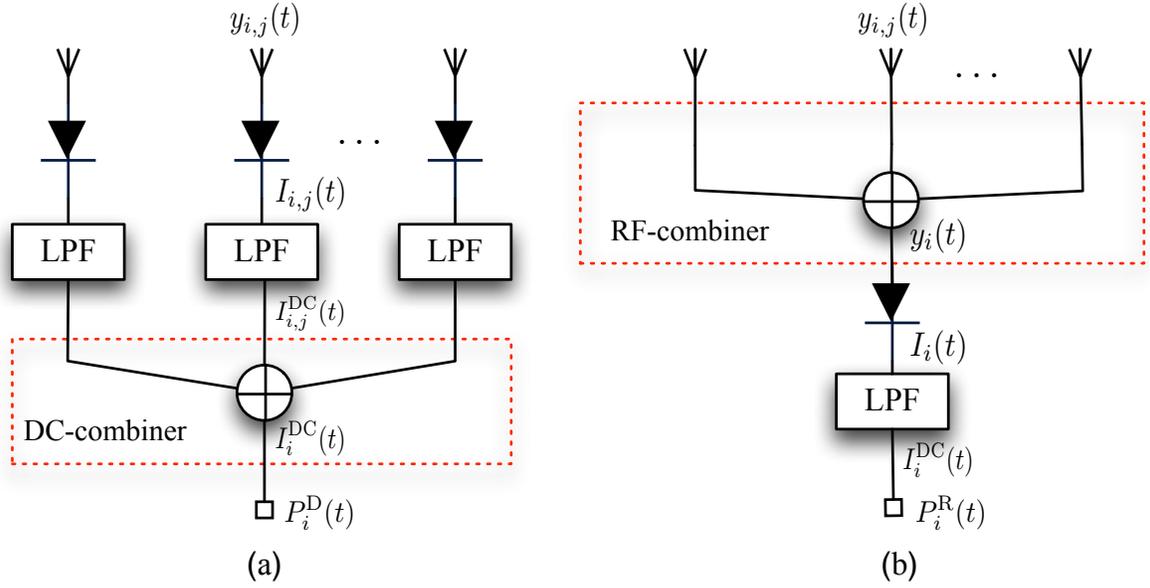}
\vspace{-1cm}
\caption{Rectenna array architectures for WPT with $L$ antenna elements; (a) DC-combiner, (b) RF-combiner.}
\label{model}
\end{figure}

\section{WPT- Rectenna array configurations}
Each terminal is equipped with an array of $L$  elements in order to boost the rectification process. The interconnection of these elements can be performed in the DC or the RF domain. 

\subsection{DC-combiner}

In the DC-combiner topology (see Fig. \ref{model}(a)), each element operates in an independent way and  has its own rectification circuit in order to harvest DC power. The DC currents at the output of each rectifier are combined (simple addition) in order to generate an aggregate DC current, which is the final input to the application.  More specifically, the output current of the Schottky diode for the $j$-th antenna element of the $i$-th terminal is given by Shockley's diode equation \cite{BOYL}  
\begin{align}
I_{i,j}(t)&=I_s\left(e^{\frac{y_{i,j}(t)}{\mu V_T}}-1 \right)=I_s \sum_{k=1}^{\infty} \left(\frac{y_{i,j}(t)}{\mu V_T}\right)^k,\label{I1}
\end{align}
where $I_s$ denotes the reverse saturation current of the diode, $\mu \in [1\;2]$ is an ideality factor (function of the operating conditions and physical contractions), and $V_T$ is the thermal voltage; the RHS in \eqref{I1} is based on the Taylor series expansion of an exponential function. The low pass filter (LPF) at the output of each diode eliminates the harmonic terms ($k>2$) in order to produce a relatively smooth DC current i.e.,  
\begin{align}
I_{i,j}^{\text{DC}}(t)=\frac{I_s}{(\mu V_T)^2} d_i^{-\alpha}P_h |g_{i,j}(t)|.
\end{align}
The DC combiners adds together the $L$ DC components in order to produce a total DC current given by
\begin{align}
I^{\text{DC}}_i(t)=e_d\sum_{j=1}^{L}I_{i,j}^{\text{DC}}(t)=\frac{e_d I_s P_h d_i^{-\alpha}}{( \mu V_T)^2} \sum_{j=1}^{L}|g_{i,j}(t)|,
\end{align}
where $e_d$ denotes the efficiency of the DC combining circuit \cite{VOL}.  The total harvested DC power is a linear function of $I^{\text{DC}}_i(t)$ and is written as 
\begin{align}
P_i^{\text{D}}(t)=\zeta_{d} I^{\text{DC}}_i(t)=\frac{\zeta_d e_dI_s P_h d_i^{-\alpha}}{(\mu V_T)^2} \sum_{j=1}^{L}|g_{i,j}(t)|^2,
\end{align}
where $\zeta_d$ denotes the conversion efficiency for the DC-combiner topology; we define $G\triangleq \frac{(\mu V_T)^2 \sigma^2}{\zeta_d e_d I_s P_h}$. 

\subsection{RF-combiner}

The RF-combiner topology (see Fig. \ref{model}(b)) requires only one rectification circuit and combines the antenna inputs in the RF domain. This combination does not requires any intelligence and passively adds the $L$ signals before the rectification process i.e.,   
\begin{align}
y_i(t)&=e_r \sum_{j=1}^{L}y_{i,j}(t) \nonumber \\
&=\sqrt{2P_h d_i^{-\alpha}} e_r \Re\left\{ e^{\jmath [2\pi ft+ \arg{x(t)}]}\sum_{j=1}^L|g_{i,j}(t)|e^{j\theta_{i,j}(t)}  \right\} \nonumber \\
&= \sqrt{2P_hd_i^{-\alpha}}|c_i(t)|e_r \cos\left(2\pi ft+ \arg{x(t)}+\theta_i(t) \right),
\end{align}
where $e_r$ denotes the efficiency of the RF combining circuit \cite{VOL}, and $\sum_{j=1}^L|g_{i,j}(t)|e^{\jmath \theta_{i,j}(t)}=|c_i(t)|e^{j\theta_i(t)}$ is a circularly-symmetric complex Gaussian random variable with zero mean and variance $L$. This combined signal is the input to the single Schottky diode;  by removing the harmonic terms ($k>2$) through LPF, the DC output and the associated harvested DC power are given by
\begin{align}
&I_i^{\text{DC}}(t)=\frac{e_r I_s P_h d_i^{-\alpha}}{(\mu V_T)^2} |c_i(t)|, \\
&P^{\text{R}}_{i}(t)=\zeta_r I_i^{\text{DC}}(t)=\frac{\zeta_r e_r I_s P_h d_i^{-\alpha}}{(\mu V_T)^2}|c_i(t)|^2,
\end{align}
where $\zeta_r$ is the conversion efficiency for the RF-combiner topology; we define $Y\triangleq \frac{(\mu V_T)^2 \sigma^2}{\zeta_r e_r I_s P_h}$.

\section{OBF with $1$-bit feedback}

The main problem in the OBF scheme is the assignment of the $M$ beams to $M$ selected terminals; this assignment requires a feedback of the achieved SINR for each beam from all terminals. Based on the received feedback, the AP allocates each beam to the terminal with the highest SINR in order to maximize the sum rate. However, this feedback channel refers to the transmission of a  high amount of information, which requires significant system resources such as bandwidth and power. In the considered WPC setup, the terminals have not their own power supply and harvest energy from the PB's RF signal in order to power their feedback transmission. Due to the small efficiency of the WPT process and the associated doubly near-far problem, we  squeeze the feedback channel into $1$-bit \cite{SIM}. In this case, the main steps of the OBF scheme are summarized as follows:
\begin{itemize}
\item The AP broadcasts the $M$ beamforming vectors to the terminals. 
\item Each terminal measures the SINR for only one beam, randomly assigned to it by the AP. The measured SINR $\Gamma_{i,l}$ is compared to the target SINR value $\Delta$ (function of the quality-of-service); $1$-bit represents whether or not $\Gamma_{i,l}$ is above the threshold.   
\item  Each terminal harvests energy through the rectenna array configuration. 
If the harvested energy is sufficient (i.e., uplink is not in outage), the terminal  transmits $1$-bit feedback to the AP, otherwise remains inactive. 
\item{Based on the received feedback, the AP randomly assigns each beam to one terminal among those who have signaled a SINR on the corresponding beam  above the threshold. If none terminal returns a positive feedback for a specific beam, assignment is performed randomly.}
\end{itemize}           

\subsection{Beam outage probability}

A beam is in outage when the achieved SINR is less than a target SINR $\Delta$; we study the outage performance of the $l$-th beam,  without loss of generality.   
In order to derive the outage probability, we firstly state the following proposition. 

\begin{proposition}
The terminals participating in the beam selection of the $l$-th beam form a homogeneous PPP with an intensity
\begin{align}
&\lambda_l^{\text{D}}=\frac{\lambda}{M(\rho^2-\xi^2)}\!\!\sum_{m=0}^{L-1}\!\! \frac{\Gamma \left(m+\frac{1}{\alpha},\!G\xi^{2\alpha} \right)\!-\!\Gamma \left(m+\frac{1}{\alpha},\!G\rho^{2\alpha} \right)}{\alpha G^{\frac{1}{\alpha}} \Gamma(m+1)}, \\
&\lambda_l^{\text{R}}=\frac{\lambda L^{\frac{1}{\alpha}}\left[\Gamma \left(\frac{1}{\alpha},\!\frac{Y}{L}\xi^{2\alpha} \right)-\Gamma \left(\frac{1}{\alpha},\!\frac{Y}{L}\rho^{2\alpha} \right) \right]}{M\alpha(\rho^2-\xi^2)Y^{\frac{1}{\alpha}}},
\end{align}
for the DC-combiner and the RF-combiner, respectively. 
\end{proposition}

\begin{proof}
Since  each terminal is preselected to observe a specific beam in a random way, the terminals which observe the $l$-th beam form a PPP $\Phi_l$ with intensity $\lambda_l=\lambda/M$ (thinning operation \cite{HAE}).   

A terminal becomes active  only when the harvested energy at the $t$-th time slot ensures the successful decoding of $1$-bit information in the uplink. This means that the Shannon capacity of the uplink between the terminal and the AP is higher than $1$ bits per channel use. For the $i$-th terminal, this condition is expressed as
\begin{align}
\log_2 \left(1+\frac{P_i^{\mathcal{Q}}(t)}{d_i^{\alpha}\sigma^2} \right)\geq 1 \Rightarrow P_i^{\mathcal{Q}}(t)\geq d_i^{\alpha} \sigma^2, 
\end{align}
where $\mathcal{Q}\in \{\text{D,R} \}$.  For the DC-combiner topology, the probability that the $i$-th terminal is idle can be written as
\begin{subequations}
\begin{align}
\Pi^{\text{D}}&=\mathbb{P}\{P_i^{\text{D}}(t)< d_i^{\alpha}\sigma^2 \}\;=\mathbb{P}\left\{Z< G d_i^{2\alpha}  \right\} \nonumber \\
&=\mathbb{E}\left[ \frac{\gamma(L,G d_i^{2\alpha})}{\Gamma(L)}\right]\;=\int_0^{2\pi} \int_{\xi}^{\rho}\frac{\gamma(L,G r^{2\alpha})}{\Gamma(L)}f_d(r) r drd\theta \nonumber \\
&=1-\frac{2}{\rho^2-\xi^2}\sum_{m=0}^{L-1} \frac{G^m}{\Gamma(m+1)} \int_{\xi}^{\rho}e^{-Gr^{2\alpha}}r^{2\alpha m+1}dr \label{step1}\\
&=1\!-\!\frac{1}{\alpha(\rho^2-\xi^2)G^{\frac{1}{\alpha}}}\sum_{m=0}^{L-1} \frac{\Gamma \left(m+\frac{1}{\alpha},\!G\xi^{2\alpha} \right)\!-\!\Gamma \left(m+\frac{1}{\alpha},\!G\rho^{2\alpha} \right)}{\Gamma(m+1) }, \label{step2}
\end{align}  
\end{subequations}
\noindent 
where $Z\triangleq \sum_{j=1}^{L}|g_{i,j}|^2$ is a central chi-square random variable with $2L$ degrees of freedom; the cumulative distribution function (CDF) of $Z$ is $F_{Z}(x)=\gamma(L,x)/\Gamma(L)$, where $\gamma(a,x)$ denotes the lower incomplete gamma function \cite{GRAD} and $\Gamma(x)$ denotes the Gamma function; $f_d(x)=1/\pi(\rho^2-\xi^2)$ denotes the probability density function of each point in $\mathcal{B}$;   \eqref{step1} is based on  \cite[8.352]{GRAD} and \eqref{step2} uses the expressions in \cite[3.381]{GRAD}. 

For the RF-combiner topology, the probability that the $i$-th terminal is not able to return a feedback is written as
\begin{align}
\Pi^{\text{R}}&=\mathbb{P}\{P_i^{\text{R}}(t)<d_i^{\alpha}\sigma^2 \}\;=\mathbb{P}\{Z_1 <Yd_i^{2\alpha}\} \nonumber \\
&=\mathbb{E}\left[1-e^{-\frac{Yd_i^{2\alpha}}{L}} \right]\;=1-\int_{0}^{2\pi}\int_{\xi}^{\rho}e^{-\frac{Y}{L}r^{2\alpha}}f_{d}(r)rdrd\theta \nonumber \\
&=1-\frac{L^{\frac{1}{\alpha}}}{\alpha(\rho^2-\xi^2)}\frac{\Gamma \left(\frac{1}{\alpha},\!\frac{Y}{L}\xi^{2\alpha} \right)-\Gamma \left(\frac{1}{\alpha},\!\frac{Y}{L}\rho^{2\alpha} \right)}{Y^{\frac{1}{\alpha}}},
\end{align}
where $Z_1\triangleq |c_i(t)|^2$ is an exponential random variable with parameter $1/L$ and  a CDF equal to $F_{Z_1}(x)=1-e^{-x/L}$.

By using thinning transformation, the terminals which feed $1$-bit for the $l$-th beam form a homogeneous PPP $\Phi_l^{\mathcal{Q}}$  with intensity $\lambda_l^{\mathcal{Q}}=\lambda_l(1-\Pi^{\mathcal{Q}})$.  Substituting the above expressions, the proposition is then proved.   
\end{proof}

According to the proposed OBF scheme, for the $l$-th beam, an outage event occurs when none terminal returns a positive feedback for the achieved beam SINR and the link (if any) between the AP and the random selected terminal is in outage. The case where all terminals return a negative feedback is equivalent to the case where the maximum  achieved beam SINR (among the terminals which feed $1$-bit) is lower than the requested threshold. For the beam outage probability, we state Theorem \ref{th1}.

\begin{theorem}\label{th1}
The beam outage probability for the $l$-th beam and the $\mathcal{Q}$ combiner ($\mathcal{Q}\in\{\rm{D,R}\}$) is given by
\begin{align}
&P_{\rm{out}}^{\mathcal{Q}}=e^{-\lambda_l^{\mathcal{Q}}\pi (\rho^2-\xi^2)(1-F_{\Gamma}(\Delta))}\bigg[ e^{-(\frac{\lambda}{M}-\lambda_l^{\mathcal{Q}})\pi (\rho^2-\xi^2)}+\left(1-e^{-(\frac{\lambda}{M}-\lambda_l^{\mathcal{Q}})\pi (\rho^2-\xi^2)} \right)F_{\Gamma}(\Delta)\bigg], 
\end{align}
where $F_{\Gamma}(x)$ is defined in \eqref{step3}.
\end{theorem}
\begin{proof}
The CDF of the observed beam SINR for a given path-loss value $d$ is written as \cite{PUN}
\begin{align}
F_{\Gamma}(x|d_i=d)=\left[1-\frac{e^{-\frac{xd^\alpha}{P_t}}}{(x+1)^{M-1}} \right]^N,
\end{align}
with expectation over $d$, the observed beam SINR has a CDF given by
\begin{align}
F_{\Gamma}(x)&=\int_{0}^{2\pi}\int_{\xi}^{\rho}\left[1-\frac{e^{-\frac{xr^\alpha}{P_t}}}{(x+1)^{M-1}} \right]^N f_d(r)r dr d\theta \nonumber \\
&\;\;=1+\frac{2}{\alpha(\rho^2-\xi^2)} \left(\frac{P_t}{x} \right)^{\frac{2}{\alpha}} \sum_{m=1}^{N}\binom{N}{m}\frac{(-1)^{m}}{m^{\frac{2}{\alpha}}(x+1)^{m(M-1)}} \nonumber \\
&\;\;\;\;\times \left[\Gamma\left(\frac{2}{\alpha},\frac{xm\xi^\alpha}{P_t}\right)-\Gamma\left(\frac{2}{\alpha},\frac{xm\rho^\alpha}{P_t}\right)\right], \label{step3}
\end{align}
where \eqref{step3} uses the binomial theorem as well as the expressions in \cite[3.381]{GRAD}. On the other hand, let's assume that $n$ terminals are able to return a feedback to the AP. By using high order statistics,  the CDF of the maximum beam SINR i.e., $\Gamma_l^*=\max_{1\leq i\leq n} \Gamma_{i,l}$,  conditioning on $n$ and path-loss values is written as  
\begin{align}
F_{\Gamma^*}(x|n,d_1,\ldots,d_n)=\prod_{i=1}^{n}F_{\Gamma}(x|d_i).
\end{align}
\noindent With expectation over the path-loss values, we have
\begin{align}
F_{\Gamma^*}(x|n)\!\!=\!\!\left[\int_{0}^{2\pi}\!\!\!\!\int_{\xi}^{\rho} F_{\Gamma}(x|r)f_d(r)d r d\theta\right]^{n}\!\!\!=[F_{\Gamma}(x)]^n.
\end{align}
The beam outage probability for the $l$-th beam and the $\mathcal{Q}$ combiner is expressed as
\begin{align}
P_{\text{out}}^{\mathcal{Q}}&=\mathbb{E}\{\Gamma_l^*<\Delta| \Phi_l^{\mathcal{Q}} \}\mathbb{E}\{\Gamma_l<\Delta| \overline{\Phi_l^{\mathcal{Q}}} \} \nonumber \\
&=\sum_{n=0}^{\infty}\mathbb{P}\{ \Gamma^*_l<\Delta|n\}\mathbb{P}\{N_{\Phi_l^{\mathcal{Q}}}(\mathcal{B})=n \} \nonumber \\
&\times \left[\mathbb{P}\{N_{\overline{\Phi_l^{\mathcal{Q}}}}(\mathcal{B})=0\}+\mathbb{P}\{N_{\overline{\Phi_l^{\mathcal{Q}}} }(\mathcal{B})>0\}\mathbb{P}\{\Gamma_l<\Delta\} \right],\label{genexp}
\end{align}
with
\begin{align}
&\mathbb{P} \big\{N_{\overline{\Phi_l^{\mathcal{Q}}}}(\mathcal{B})=0 \big\}=e^{-\overline{\lambda_l^{\mathcal{Q}}}|\mathcal{B}|}, \label{exp1}\\
&\sum_{n=0}^{\infty}\mathbb{P}\{ \Gamma^*_l<\Delta|n\}\mathbb{P}\{N_{\Phi_l^{\mathcal{Q}}}(\mathcal{B})=n \} \nonumber \\
&=\sum_{n=0}^{\infty}\frac{e^{-\lambda_l^{\mathcal{Q}}|\mathcal{B}|}\left(\lambda_l^{\mathcal{Q}}|\mathcal{B}|F_{\Gamma}(\Delta) \right)^n}{n!}=e^{-\lambda_l^{\mathcal{Q}}|\mathcal{B}|(1-F_{\Gamma}(\Delta))}, \label{exp2}
\end{align}
where $\overline{\Phi_l^{\mathcal{Q}}}$ denotes  the complementary homogeneous PPP with intensity $\overline{\lambda_l^{\mathcal{Q}}}\triangleq \lambda_l \Pi^{\mathcal{Q}}$, which is formed by the terminals that are not able to return $1$-bit feedback; $N_{\mathcal{Z}}(\mathcal{B})$ denotes the number of terminal in $\mathcal{B}$ for a PPP $\mathcal{Z}$, and $|\mathcal{B}|=\pi(\rho^2-\xi^2)$. By substituting  \eqref{step3}, \eqref{exp1}, \eqref{exp2} into \eqref{genexp}, we prove the statement. 
\end{proof}

\begin{remark}\label{rm1}
For $P_t,P_h\rightarrow \infty$, the beam outage probability for both combination schemes asymptotically converges  to 
\begin{align}
P_{\rm{out}}=e^{-\frac{\lambda}{M} \pi (\rho^2-\xi^2)(1-F_{\Gamma}^{\infty}(\Delta))},
\end{align}
where $F_{\Gamma}^{\infty}(x)=[1-1/(x+1)^{M-1}]^N$. 
\end{remark}
\begin{proof}
For $P_t,P_h\rightarrow \infty$, all the terminals successfully transmit $1$-bit feedback and thus $\lambda_l^{\mathcal{Q}}\rightarrow \lambda_l$. 
Remark \ref{rm1}
can be straightforwardly obtained from $P_{\text{out}}=\mathbb{E}\{\Gamma_l^*<\Delta| \Phi_l \}$.
\end{proof}

\section{Numerical results and discussion}

\begin{figure}[t]
\centering
\includegraphics[width=0.8\linewidth]{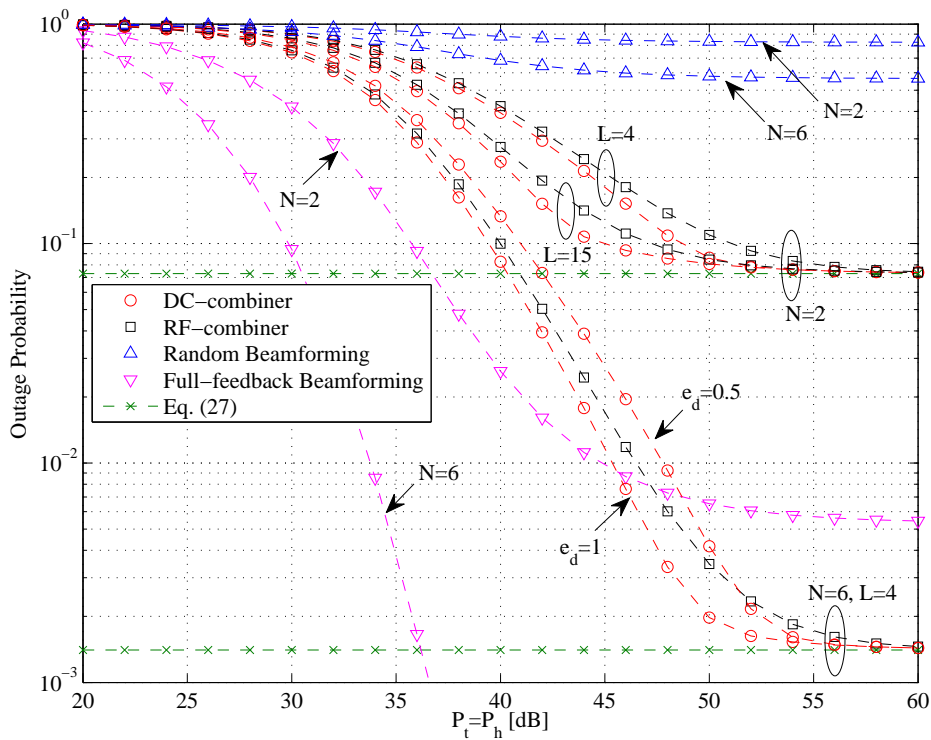}
\vspace{-0.7cm}
\caption{Outage Probability versus the transmit power with $P_t=P_h$; $\xi=2$ m, $\rho=10$ m, $I_s=1$ mA, $\mu=2$, $V_T=28.85$ mV, $\zeta_d=\zeta_r=0.9$, $e_r=1$, $e_d=\{1,0.5\}$, $\sigma^2=-10$ dBm, $\alpha=3$, $M=2$, $N=\{2,6\}$, $L=\{4,15\}$, $\Delta=10$ dB, and $\lambda=0.1$; the dashed lines represent the theoretical results.}
\label{fig1}
\end{figure}

Fig. \ref{fig1}  plots the beam outage probability versus the transmit power for both RF/DC combination schemes; the random beamforming (no feedback) and the full-feedback OBF \cite{PUN} (terminals feed SINR information for all beams)  are used as benchmarks. The first main observation is that the proposed $1$-bit OBF significantly outperforms random assignments. The associated gain increases as the number of receive antennas $N$ increases (receive diversity).  On the other hand, it can be seen that the beam outage probability performance is improved as $L$ increases for moderate values of $P_h$ (see the case with $N=2$). As the size of the rectenna array increases, the terminals harvest more energy and therefore return $1$-bit feedback with a higher probability.

As for the combination topologies, it can be seen that the DC-combiner outperfoms the RF-combiner for the ideal case with $e_d=e_r=1$ e.g., gain of $3$ dB for $P_t=50$ dB and $N=6$. However, the performance of the combination schemes and their suitability highly depend on the quality of the associated combining circuits. We can  observe that the RF-combiner achieves a lower outage probability than the DC-combiner when $e_r=1$ and $e_d=0.5$ (see the case with $N=6$). Therefore, the designer should carefully take into account the insertion losses of the combining circuits in order to select the best configuration. For high values of $P_t,P_h$ both combination schemes converge to the same outage probability floor, in accordance to Remark \ref{rm1}.

\end{document}